\documentclass[11pt]{amsart}
\usepackage{amsfonts}
\usepackage{ifthen}
\usepackage{amsthm}
\usepackage{amsmath}
\usepackage{graphicx}
\usepackage{amscd,amssymb,amsthm}
\usepackage{enumerate}
\usepackage{hyperref}
\newcounter{minutes}
\divide\time by 60
\newcounter{hours}
\multiply\time by 60 \addtocounter{minutes}{-\time}

\title[Generalized pathway entropy and its applications ...]{Generalized pathway entropy and its applications in diffusion entropy analysis and fractional calculus}

\newtheorem{theorem}{Theorem}
\newtheorem{remark}{Remark}
\newtheorem{lemma}{Lemma}
\newtheorem{corollary}{Corollary}
\newtheorem{example}{Example}

\begin{document}

\def\thefootnote{}
\footnotetext{ \texttt{File:~\jobname .tex,
          printed: \number\year-0\number\month-\number\day,
          \thehours.\ifnum\theminutes<10{0}\fi\theminutes}
} \makeatletter\def\thefootnote{\@arabic\c@footnote}\makeatother

\maketitle


\begin{center}

\author{NICY SEBASTIAN}\\
\address{Indian Statistical Institute, Chennai Centre, SETS Campus, Taramani, Chennai\\ India-600 113.}\\
\email{nicy@isichennai.res.in; nicyseb@yahoo.com}\\
\end{center}

\begin{abstract}
We presented background information about various entropies in the literature.
The pathway idea of Mathai (2005) is shown to be inferable
from the maximization of a certain generalized entropy measure and established connections to outstanding problems in astronomy and physics.
In this paper we proved that the generalized entropy of Mathai associated with diffusive processes
grows linearly with the logarithm of time, and the rate of growth is independent of the generalized entropy parameter.
 We also proposed some results concerning images of generalized Bessel function under the pathway operator and its special cases including some trigonometric functions. Situations are listed where a generalized entropy of order $\alpha$ leads to
pathway models, exponential and power law behavior and related differential
equations.
\end{abstract}

\section{Introduction.}

The normal (Gaussian) distribution is a family of continuous density functions
and is ubiquitous in the field of statistics and probability (Feller \cite{Feller}). The importance of
the normal distribution as a model of quantitative phenomena is due to the central limit
theorem. The normal distribution maximizes Shannon entropy among all distributions
with known mean and variance and in information theory, Shannon entropy is the measure
of uncertainty associated with a random variable.

In statistical mechanics, Gaussian (Maxwell-Boltzmann) distribution maximizes the
Boltzmann-Gibbs entropy under appropriate constraints (Gell-Mann and Tsallis \cite{Gell}).
Given a probability distribution $P = \{p_i\} (i = 1,\ldots,m),$ with $p_i$ representing the
probability of the system to be in the $i$th microstate, the Boltzmann-Gibbs entropy is
\begin{equation}\label{eq:0}S=-k\sum_{i=1}^mp_i \ln p_i, \end{equation}
 where $k$ is the Boltzmann constant and $m$ the total number
of microstates. If all states are equally probable it leads to the Boltzmann principle
$S = k \ln W (m = W)$. Boltzmann-Gibbs entropy is equivalent to Shannon’s entropy if
$k = 1$.
If we consider such a system in contact with a thermostat
then we obtain the usual Maxwell-Boltzmann distribution for the possible
states by maximizing the Boltzmann-Gibbs entropy $S$ with the normalization
and energy constraints. However, in nature many systems show distributions
which differ from the Maxwell-Boltzmann distribution. These are
usually systems with strong autocorrelations preventing the convergence to
the Maxwell-Boltzmann distribution in the sense of the central-limit theorem.
Well known examples in physics are: self gravitating systems, charged
plasmas, Brownian particles in the presence of driving forces, and, more generally,
non-equilibrium states of physical systems (Abe and Okamoto \cite{Abe},
Gell-Mann and Tsallis \cite{Gell}). Then it is natural to ask the question of
whether non-Maxwell-Boltzmannian distributions can also be obtained from
a corresponding maximum entropy principle, considering a generalized form
for the entropy. For this purpose, different forms were proposed, as for instance
we are investigating the link between entropic functionals and
the corresponding families of distributions in Mathai's pathway model and we
come to the conclusion that this link is also important to physically analyze
fractional reaction equations in terms of probability theory.

The entropy is the rigorous measure of lack of information. 
%
%
%
%
%
%
%
%
%
%
%
%
%
%
%
The following are some of the generalizations of Shannon $S$.
\begin{equation}\label{eq:1}R= \frac{ \ln(\sum_{i=1}^k p_i^{\alpha })}{1-\alpha}, \; \alpha \neq 1,\alpha >0,(\text{R\'{e}nyi entropy of order}~ \alpha ~\text{of 1961}),\end{equation}
\begin{equation}\label{eq:2}H= \frac{ (\sum_{i=1}^k p_i^{\alpha }-1)}{2^{1-\alpha}-1}, \; \alpha \neq 1,\alpha >0,(\text{Havrda-Charv\'{a}t entropy of order} ~\alpha~ \text{of 1967}),\end{equation}

\begin{equation}\label{eq:3}T= \frac{ (\sum_{i=1}^k p_i^{\alpha })}{1-\alpha}, \; \alpha \neq 1,\alpha >0,(\text{Tsallis non-extensive entropy of order}~ \alpha ~\text{of 1988}),\end{equation}
where $ p_i>0, i=1,\ldots,k, p_1+\ldots+p_k=1$. When $\alpha\rightarrow1$ all the entropies of order $\alpha$ described in (\ref{eq:1}) to (\ref{eq:3}) go to Shannon entropy $S$.

In physical situations when an appropriate density is selected, one
procedure is the maximization of entropy\index{entropy}. Mathai and Rathie \cite{Rathie}
consider various generalizations of Shannon\index{Shannon entropy}
entropy\index{entropy} measure and describe various properties
including additivity, characterization theorem etc. Mathai et al. \cite{mathaihaubold2007}
introduced a new generalized entropy\index{entropy} measure which is
a generalization of the Shannon \index{Shannon entropy}
entropy\index{entropy} measure. Applying the maximum entropy principle with
normalization and energy constraints to Mathai's entropic functional, the corresponding parametric
families of distributions of generalized type-1 beta, type-2 beta, generalized gamma, generalized
Mittag-Leffler, and L\'{e}vy are obtained. For a multinomial population $P=
(p_1, \ldots, p_k), ~p_i\geq 0,~ i=1, \ldots, k,~p_1+p_2+\cdots+p_k
=1$,  the Mathai's\index{Mathai's Entropy} entropy\index{entropy}
measure is given by the relation
\begin{equation}\label{eq:4}
M_{k,\alpha }(P)= \frac{\displaystyle\sum_{i=1}^k p_i^{2-\alpha }-1}{\alpha
-1}, \; \alpha \neq 1,\;-\infty <\alpha <2,~\text{(discrete case)}\end{equation}
\begin{equation}\label{eq:5} M_{\alpha }(f)= \frac{1}{\alpha-1}\left[\int_{-\infty}^\infty[f(x)]^{2-\alpha}{\rm d}x-1\right],~\alpha \neq1,~\alpha<2~\text{(continuous case)}.\end{equation}

By optimizing Mathai's entropy\index{entropy} measure, one can arrive at pathway model of Mathai \cite{amm2005}, which consists of  many of the standard
distributions in statistical literature as special cases. For fixed $\alpha$, consider the optimization of $M_\alpha(f)$, which implies optimization of $\int_x[f(x)]^{2-\alpha}{\rm d}x$, subject to the following conditions:
\begin{enumerate}[{(i)}]
\item $f(x)\geq 0,~\text{for all}~x$
\item $\int_x f(x){\rm d}x < \infty$
\item $\int_x x^{\rho(1-\alpha)}f(x){\rm d}x=\text{fixed for all}~f$
\item $\int_x x^{\rho(1-\alpha)+\delta}f(x){\rm d}x=\text{fixed for all}~f, \text{where}~\rho~ \text{and}~\delta ~\text{are fixed parameters.}$
\end{enumerate}
By using calculus of variation, one can obtain the Euler equation as
\begin{eqnarray}
&&\frac{\partial}{\partial f}[f^{2-\alpha}-\lambda_1x^{\rho(1-\alpha)}f+\lambda_2x^{\rho(1-\alpha)+\delta}f]=0\nonumber\\
&&\Rightarrow (2-\alpha)f^{1-\alpha}=\lambda_1x^{\rho(1-\alpha)}[1-\frac{\lambda_2}{\lambda_1}x^\delta],~\alpha\neq 1,2\nonumber\\
&&\Rightarrow f_1=c_1x^\rho[1-a(1-\alpha)x^\delta]^{\frac{1}{1-\alpha}}
\end{eqnarray}
for $\frac{\lambda_2}{\lambda_1}=a(1-\alpha), a>0 $ with $\alpha<1$ for type-1 beta, $\alpha > 1$ for type-2 beta, $\alpha\rightarrow1$ for gamma, and $\delta = 1$ for Tsallis
statistics.  For more details the reader may refer to the papers of Mathai et al.\cite{mathaihaubold2007}, Mathai and Haubold \cite{mathaihaubold2007a}.
When $\alpha \rightarrow 1$,  the Mathai's entropy measure $M_{\alpha}(f)$ goes to
the Shannon \index{Shannon entropy} entropy\index{entropy} measure
and this is a variant of
 Havrda-Charv\'{a}t entropy\index{entropy}, and the variant form therein is Tsallis
 entropy\index{entropy}. Then when $\alpha$ increases from 1, $M_{\alpha }(f)$ moves away from Shannon entropy. Thus $\alpha$ creates a pathway moving from one function to another, through the generalized entropy also. This is the entropic pathway. One can derive  Tsallis statistics and superstatistics (Beck \cite{Beck},  Beck and Cohen \cite{Cohen}) by using  Mathai's\index{Mathai's Entropy} entropy.
It is shown that when the model is applied to physical situations
then the current hot topics of Tsallis statistics and superstatistics in statistical mechanics
become special cases of the pathway model, and the model is capable of
capturing many stable situations as well as the unstable or chaotic neighborhoods
of the stable situations and transitional stages.

In Section 2 we demonstrate that the extensive generalized entropy associated with diffusive processes
grows linearly with the logarithm of time, and the rate of growth is independent of the extensive generalized entropy parameter. In Section 3 we proposed some results concerning images of generalized Bessel function under the pathway operator and its special cases including some trigonometric functions. In the last section we listed an example where a generalized entropy of order $\alpha$ leads to
pathway models, exponential and power law behavior and related differential
equations.
\section{Diffusion entropy analysis.}

In this section we focus upon the scaling properties of a time series. By
summing the terms of a time series we get a trajectory and the trajectory can be
used to generate a diffusion process. There is scaling if, in the stationary condition, a
diffusion process can be described by the following probability function (pdf):
\begin{equation}\label{eq:6}
p(x,t)=\frac{1}{t^\delta}F(\frac{x}{t^\delta}),
\end{equation}
where $x$ denotes the diffusion variable and $p(x, t)$ is its pdf at time $t$. The coefficient
$\delta$ is called the scaling exponent. We define the scaling of a time series as the scaling
exponent of a diffusion process generated by that time series. The purpose of the Diffusion Entropy Analysis (DEA) algorithm is to establish the
possible existence of scaling, either normal or anomalous, in the most efficient way
as possible without altering the data with any form of detrending.

Let us consider the simplifying assumption of considering large enough times as
to make the continuous assumption valid. This method of analysis is based upon the evaluation of the Shannon
entropy (continuous version) of the pdf of the diffusion process that reads
\begin{equation}\label{eq:7}
S(t)=-\int^{+\infty}_{-\infty} dx\;p(x,t) \ln\;p(x,t).
\end{equation}
Using the scaling condition of (\ref{eq:6}) we obtain
\begin{equation}\label{eq:8}
S(t)=A+\delta \ln\;t,\;\;A=-\int^{+\infty}_{-\infty}dyF(y)\ln F(y),
\end{equation}where $y=\frac{x}{t^\delta}$, for more details see \cite{Scafetta2010}, \cite{Scafetta2002}.
Equation (\ref{eq:8}) indicates that in the case of a diffusion process
with a scaling pdf, its entropy $S(t)$ increases linearly with $\ln \;t$. Numerically, the
scaling exponent $\delta$ can be evaluated by using fitting curves with the function
of the form $f_S(t) = \kappa +\delta \ln\;t$ that, when graphed on linear-log graph paper,
yields straight lines.

Hence the DEA provides a better way to detect $\delta$ correctly. It
is so because DEA analyzes directly the pdf of the diffusion
processes, without using the moments of the distribution. Instead, all the other
methods used for detecting scaling Variance Scaling Analysis, Hurst R/S Analysis,
Detrended Fluctuation Analysis, Relative Dispersion Analysis, Spectral Analysis,
Spectral Wavelet Analysis are subtly based on the Gaussian assumption and, so,
upon a variance that can be used to monitor scaling.
In the variance based methods, scaling is studied by direct evaluation of the time behavior of the variance of the diffusion process. If the variance scales, one would have
$$
\sigma_x^2(t)\sim t^{2H},
$$
where $H$ is the Hurst exponent in honor of Hurst \cite{Hurst}. The problem is that the scaling
detected by the variance methods, $H$, may
not exist or may not coincide with the correct scaling, $\delta$. If the time series is characterized
by what Mandelbrot called Fractional Brownian Motion, we have $H = \delta$.
Consequently, the scaling of this type of noise can be detected by using the variance
methods. If, on the contrary, the time series is characterized, for example, by L\'{e}vy
properties (\cite{Levy1}, \cite{Levy2}), $H \neq \delta$ and the variance methods cannot be used to detect the true
scaling. A diffusion process generated by L\'{e}vy walk is characterized by the relation
\begin{equation}\label{eq:9}\delta=\frac{1}{3-2H}.\end{equation}
In the case of L\'{e}vy flights, the exponent $H$ cannot be determined because the variance
diverges, whereas the scaling $\delta$ exists and can be determined by using the diffusion
entropy analysis. The above conclusions suggest that to determine the real statistical
properties of a time series it is not enough to study the scaling with only one type
of analysis. Only the joint use of two scaling analysis methods, the variance scaling
analysis and the diffusion entropy analysis, can determine the real nature, Gauss or
L\'{e}vy or something else, of a time series. We have to determine $H$ and $\delta$. Then, if
$H = \delta$ we can conclude that fractional Brownian noise may characterize the signal.
If, instead, $H \neq \delta$ we have to look for a different type of noise. If we find that
the relation (\ref{eq:9}) holds true, we can have good reasons to conclude that the noise is
characterized by L\'{e}vy statistics. Moreover, DEA may be used
for studying the transition from the dynamics to the thermodynamics of the diffusion
process.
\subsection{Mathai's entropy.}
The entropies defined in (\ref{eq:2}), (\ref{eq:3}) and (\ref{eq:4}) are non-additive. An additive form of (\ref{eq:4}) is defined as follows: For a multinomial population $P=
(p_1, \ldots, p_k), ~p_i\geq 0,~ i=1, \ldots, k,~p_1+p_2+\cdots+p_k
=1$,  the Mathai's\index{Mathai's Entropy} extensive generalized entropy\index{entropy}
measure is given by the relation
\begin{equation}\label{eq:10}
M_{k,\alpha }^{*}(P)= \frac{\ln(\sum_{i=1}^k p_i^{2-\alpha })}{\alpha
-1}, \; \alpha \neq 1,\;-\infty <\alpha <2,~\text{(discrete case).}\end{equation}
As can be expected, when a logarithmic function is involved, as in the case of (\ref{eq:0}), (\ref{eq:1}) and (\ref{eq:10}) the entropy is additive.
The DEA performs better than the other methods of analysis due to the fact that the information
extracted from the pdf, expressed under the form of entropy, is larger than
the information extracted from the pdf variance. Note that the entropy indicator
need not to be the Shannon indicator. We hope that
to detect scaling  extensive generalized entropy defined in (\ref{eq:10}) is as
effective as the Shannon entropy.

In the continuum limit: $p_i \approx p(x,t)\Delta x,$ where $x$ is the random variable, for example displacement for random walker,
and $p(\cdot)$ is the probability function. We consider Brownian and anomalous diffusive processes
characterized by a probability function given in (\ref{eq:6}). The normalization of probabilities implies:
$\int f(x) dx =1$. A Brownian process
is characterized by the lack of time correlations and has the probability function:
$$p(x,t)= \frac{1}{\sqrt{4\pi Dt}}{\rm e}^{-\frac{x^2}{4Dt}},$$
 where $D$ is the diffusion constant.
In the subdiffusive regime, $0<\delta<\frac{1}{2},$ there are negative correlations or antipersistence,
while in the superdiffusive regime (L\'{e}vy flights), $\frac{1}{2}<\delta<1,$ there are positive time correlations or
persistence.

The time dependence of entropy for anomalous diffusion processes is obtained by substituting
$p_i \approx p(x,t)\Delta x,$ into (\ref{eq:10}), replacing the sum by an integral and using (\ref{eq:6}) we have
$$M_{k,\alpha }^{*}(P,t)= \frac{\ln(\sum_{i=1}^k p_i^{2-\alpha })}{\alpha
-1}\approx\frac{1}{\alpha
-1}\ln\left[(\Delta x)^{1-\alpha}  \int\; (p(x,t))^{2-\alpha} dx    \right] .$$
Using simple algebra we will get the associated extensive generalized entropy in the following form
\begin{equation}\label{eq:11}
M_{k,\alpha }^{*}(P,t)=- \ln\; \Delta x- \frac{1}{1-\alpha}\ln \left[ \int (f(y))^{2-\alpha} dy\right]+\delta \ln t.
\end{equation}where $y=\frac{x}{t^\delta}$.
First we note that Equation (\ref{eq:11}) is precisely analogous to the logarithmic time evolution
in (\ref{eq:8}). Hence $M_{k,\alpha }^{*}(P,t)= a(\alpha) \ln \; t +b(\alpha),$ with $a(\alpha)= \delta $ independent of the extensive generalized entropy parameter $\alpha$.

The non-stationary dynamical transient may be simulated by a non stationary
pdf of the type
\begin{equation}\label{eq:11a}p(x,t)=\frac{1}{t^{\delta(t)}}F(\frac{x}{t^{\delta(t)}}),\end{equation}
where the pdf scaling exponent $\delta(t)$ changes with the diffusion time $t$. Let us suppose
that
\begin{equation}\label{eq:12}\delta(t)=\delta_0+\eta \ln\; t.\end{equation}
Since the scaling parameter $\delta$ cannot exceed the ballistic value $\delta = 1$ in the case of a
dynamical approach to diffusion with
fluctuation of limited intensity, this condition applies to the time scale defined by
$$\eta \ln\; t<1-\delta_0.$$
We notice
that in the new non-stationary condition the traditional entropy indicator yields
\begin{equation}\label{eq:13}M_{k,\alpha }^{*}(P,\tau)= b(\alpha)+ \delta_0 \tau+\eta \tau^2 ,\end{equation}
where $\tau=\ln\; t$.
Furthermore, the benefits stemming from the entropic method of analysis of a
diffusion process (the DEA) are not limited to the detection of the true asymptotic
scaling $\delta$. We can explore the still unknown regime of transition from dynamics to
thermodynamics, and we can also address the ambitious issue of studying the time
series produced by non-stationary processes.

\subsection{Diffusion Entropy Analysis Based on Non-extensive Mathai's entropy.}
A diffusion entropy analysis study based on non-extensive Tsallis indicator can be obtained in \cite{Scafetta2010}.
The non-extensive Tsallis indicator corresponding to the continuous formalism is given by
\begin{equation}\label{eq:5a} T_{\alpha }(t)= \frac{1}{\alpha-1}\left[1-\int_{-\infty}^\infty {\rm d}x[p(x,t)]^{\alpha}\right].\end{equation}

The quadratic form of (\ref{eq:13}) suggests that the choice of $\delta(t)$ given by (\ref{eq:12})
has the mathematical meaning of the quadratic term in the Taylor expansion of
the diffusion entropy (\ref{eq:7}). As a consequence, we should expect that, in general,
$\delta(t)$ always assumes the form of (\ref{eq:12}), at least for small values of $\ln \;t$.
The non-extensive Mathai's entropy \cite{mathaihaubold2007} reads in the continuous formalism
\begin{equation}\label{eq:5b} M_{\alpha }(t)= \frac{1}{\alpha-1}\left[\int_{-\infty}^\infty {\rm d}x [p(x,t)]^{2-\alpha}-1\right],~\alpha \neq1,~\alpha<2~\text{(continuous case)}.\end{equation}
It is straightforward to prove that this entropic indicator coincides with that of
(\ref{eq:7}) in the limit where the entropic index $\alpha \rightarrow 1$. Let us make the assumption
that in the diffusion regime, the departure from this traditional value is weak,
and assume $\epsilon \equiv \alpha-1 \ll 1 $. This allows us to use the following approximate
expression for the non-extensive entropy
$$M_{\alpha}(t)= -\int^{+\infty}_{-\infty} dx\;p(x,t) \ln\;p(x,t)+ \frac{\epsilon}{2}\int^{+\infty}_{-\infty} dx\;p(x,t) [\ln\;p(x,t)]^2.$$
In the specific case where the nonscaling condition of (\ref{eq:11a}) applies, this entropy yields the form
$$M_{\alpha}= A+\epsilon B +(1+\epsilon A)\delta (t) \ln\; t+ \frac{\epsilon}{2}[\delta (t) \ln\; t]^2,$$ where $A$ and $B$ are two constants related to $F(y)$ of (\ref{eq:11a}). These theoretical remarks demonstrate that this non-extensive approach to the
diffusion entropy makes it possible to detect the strength of the deviation from the
steady condition. In fact it could be proves that $\epsilon = 0$ implies a steady condition. The
conclusion of this section is that the breakdown of the scaling property of (\ref{eq:6})
can be revealed by the DEA under the form of an entropic index $\alpha$ departing
from the condition of ordinary statistical mechanics, namely $\alpha = 1$.

\begin{figure}[h!]
\begin{center}
~~~~~ \resizebox{7cm}{7cm}{\includegraphics{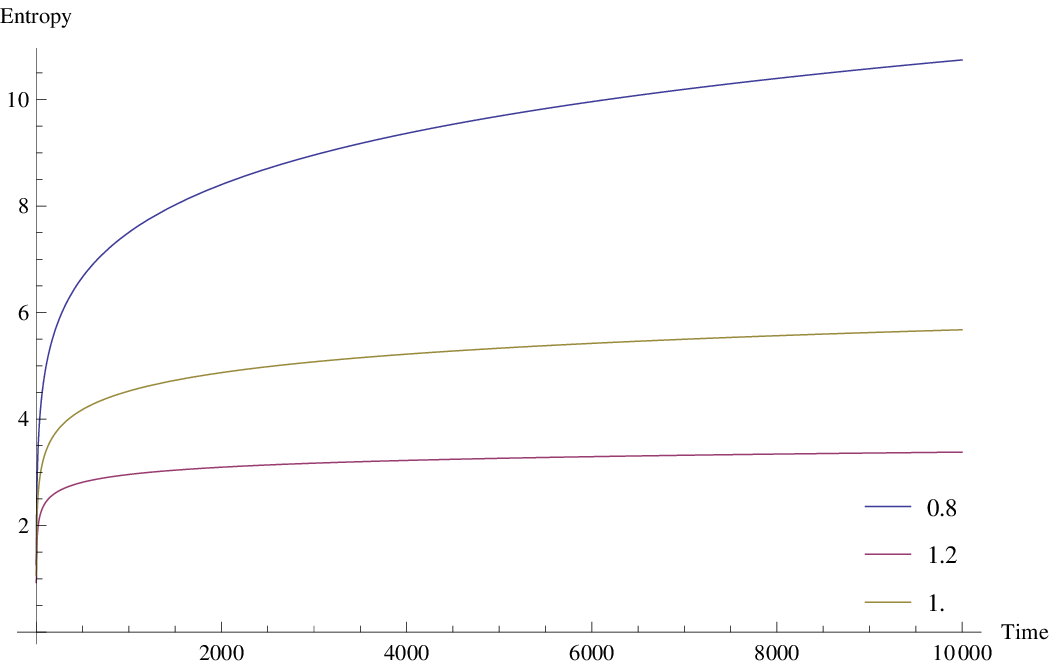}}
~~~~ \resizebox{7cm}{7cm}{\includegraphics{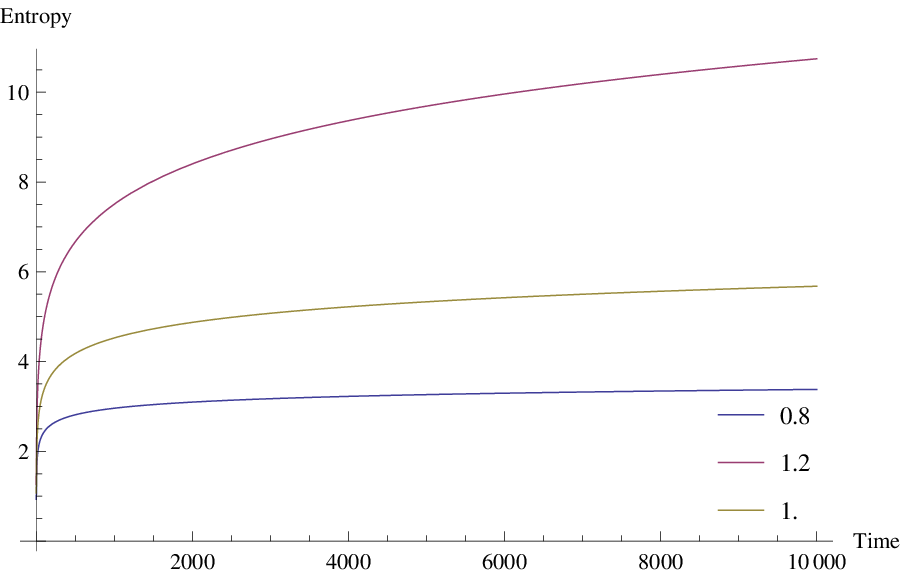}}\\
\caption {Diffusion Entropy by using the non-extensive (a) Tsallis entropy  (\ref{eq:5a})
\hskip1.5cm (b) Mathai's entropy (\ref{eq:5b}) \label{plot1}}
\end{center}
\end{figure}

Figure \ref{plot1}(a) and (b) show the curves corresponding to $\alpha=0.8,1,1.2$ respectively with the non-extensive Tsallis and Mathai's $\alpha$-entropy indicator as a function of time $t$ applied to the following Brownian diffusion equation
$$p(x,t)=\frac{1}{\sqrt{\pi t}}{\rm e}^{-\frac{x^2}{t}}.$$
By adopting the non-extensive Tsallis entropy (\ref{eq:5a}), we get
\begin{equation}\label{eq:2a}
 S_{\alpha}(t)=\left\{\begin{array}{ll}
\frac{(1-\pi^{\frac{1-\alpha}{2}}\alpha^{-\frac{1}{2}}t^{\frac{1-\alpha}{2}})}{\alpha-1}
 & \alpha \neq 1\\
\frac{1}{2}+\frac{1}{2} \ln\;(\pi t) & \alpha=1,
 \end{array} \right .
\end{equation}
and by using the non-extensive Mathai's entropy (\ref{eq:5b}), we get
\begin{equation}\label{eq:2b}
 M_{\alpha}(t)=\left\{\begin{array}{ll}
\frac{(\pi^{\frac{\alpha-1}{2}}(2-\alpha)^{-\frac{1}{2}}t^{\frac{\alpha-1}{2}}-1)}{\alpha-1}
 & \alpha \neq 1, \alpha<2\\
\frac{1}{2}+\frac{1}{2} \ln\;(\pi t) & \alpha=1.
 \end{array} \right .
\end{equation}

\section{Pathway integral operator of generalized Bessel function and their special cases.}

By using the pathway idea of Mathai \cite{amm2005}, a pathway fractional
integral operator (pathway operator) is defined by Nair \cite{nair} and is defined as follows:
Let $f(x)\in L(a, b), \eta\in C, \Re(\eta)>0, a>0$
and $\alpha<1,$ then
\begin{equation}\label{eq:4.1}
(P_{0+}^{(\eta,\alpha)}f)(x)=x^{\eta-1}\int_0^{\frac{x}{a(1-\alpha)}}\left[1-\frac{a(1-\alpha)t}{x}\right]^{\frac{\eta}{(1-\alpha)}-1} f(t){\rm d}
t,\end{equation}
where $\alpha$ is the pathway parameter and $f(t)$ is an arbitrary function.
In the pathway model, as $\alpha\rightarrow1$, we can see that
$$
\lim_{\alpha\rightarrow1_{-}}[1-a(1-\alpha)x^{\delta}]^{\frac{\eta}{1-\alpha}}
=
\lim_{\alpha\rightarrow1_{+}}[1+a(\alpha-1)x^{\delta}]^{-\frac{\eta}{\alpha-1}}
   = {\rm e}^{-a\eta x^{\delta}}.
$$
When $\alpha\rightarrow  1_{-}, [1-\frac{a(1-\alpha)t}
{x} ]^{\frac{\eta}{1-\alpha}}
\rightarrow {\rm e}^{-\frac{a\eta}{x}t}$. Thus the operator will become
$$P_{0+}^{\eta,1}=x^{\eta-1}\int_0^\infty{\rm e}^{\frac{-a\eta}{x}t} f(t){\rm d} t=x^{\eta-1}L_f(\frac{a\eta}{x}), $$
the Laplace transform of $f$ with parameter $\frac{a\eta}{x}$. When $\alpha = 0, a = 1$ in (\ref{eq:4.1}) the integral
will become,
$$\int_0^x(x-t)^{\eta-1}f(t){\rm d} t=\Gamma(\eta)I_{0+}^\eta,$$
where $I_{0+}$ is the left-sided Riemann-Liouville fractional integral operator defined for  $\eta \in \mathbb{C},
 ~x>0,$ (Samko et al, \cite{SKM}) as:
  \begin{equation}
(\mathcal{I}_{0+}^{\eta}f)(x)
 \equiv\frac{1}{\Gamma(\eta)}\int_o^x (x-t)^{\eta-1}f(t)dt\ \
  ~\Re(\eta)>0.
  \end{equation}
It is also observed that when the pathway parameter, $\alpha = 0,a = 1$ and $f(t)$ replaced by $_2F_1\left(\nu+\beta, -\eta;\nu;1-\frac{t}{x}\right)f(t)$ then the pathway operator yields to $$\int_0^x (x-t)^{\alpha-1}\
_2F_1\left(\nu+\beta, -\eta;\nu;1-\frac{t}{x}\right)f(t){\rm
d}t=\frac{\Gamma(\nu)}{x^{-\nu-\beta}}I_{0+}^{\nu, \beta, \eta},$$
where $I_{0+}^{\nu, \beta, \eta}$ denotes the Saigo fractional integral operator, \cite{S}.
When $ \alpha\rightarrow1, \eta=1$ and replace $f(t)$ by $t^{\beta-1 }{_0F_1}(~; \beta;\delta t)$ in pathway fractional integral operator then we are essentially dealing with
distribution functions under a gamma Bessel type
model in a practical statistical problem, see \cite{Sebastian}.
 Hence a connection between statistical distribution
theory and fractional calculus is established so that one can make use of the rich results in
statistical distribution theory for further development of fractional calculus and
vice versa.
 The pathway fractional
integral operator has found applications in reaction-diffusion problems, non-extensive
statistical mechanics, non-linear waves, fractional differential equations, non-stable
neighborhoods of physical system etc.

Our goal is to study in general the pathway fractional integration of the Bessel functions, the modified Bessel functions, the
spherical Bessel functions and the modified spherical Bessel functions together. For
this we consider the linear differential equation
 \begin{equation}\label{eq5}z^2w'' (z)+bzw'(z)+(cz^2+d)w(z)=0, \end{equation}
where $b,c \in \mathbb{C}$ and $d = d_1p^2+d_2p+d_3,$ with $d_1, d_2, d_3, p \in \mathbb{C}$. By putting $d_1 = -1, d_2 = 1-b $ and $d_3 = 0$, and the differential equation (\ref{eq5}) is
\begin{equation}\label{eq6}z^2w'' (z)+bzw'(z)+(cz^2-p^2+(1-b)p)w(z)=0. \end{equation}
We obtain a particular solution
of (\ref{eq6}) for all $z\in \mathbb{C}, z\neq0, $  and $b,c, p \in\mathbb{C}$ by \cite{baricz1},
\begin{equation}\label{eq:a}
W_{p,b,c}(z)=\sum_{k=0}^\infty\frac{(-c)^k}{k!\Gamma(p+\frac{b+1}{2}+k)}\left(\frac{z}{2}\right)^{\nu+2k},
\end{equation}
where $W_{p,b,c}(\cdot)$ is the generalized Bessel function of the first kind and which permits
the study of Bessel, modified Bessel, spherical Bessel and modified spherical Bessel
functions together. It is clear that for $c=1$ and $b=1$ the function
$W_{p,b,c}$ reduces to $J_p,$ Bessel function of the fist kind of order $p$, when $c=-1$ and
$b = 1$ the function wp becomes $I_p$, is modified Bessel function of the fist kind of order p. Similarly, when $c = 1$ and $b = 2$ the function $W_{p,b,c}$
reduces to $2 j_p/\sqrt{\pi}$ where $j_p$ is the spherical
Bessel function of order $p$, while if $c = -1$ and $b = 2$, then $W_{p,b,c}$ becomes $2i_p/
\sqrt{\pi}$, where $i_p$ is the modified spherical Bessel function of
order $p$. Further,
from (\ref{eq:a}) we have $W_p(0) = 0$.

Our main result in this section is based on the preliminary
assertion giving composition formula of pathway fractional
integral operator (\ref{eq:4.1}) with a power
function.
\begin{lemma}
\label{lem1}[\cite{nair}, Lemma 1]   Let $\rho,\eta\in \mathbb{C}, \Re(\eta)>0  $ and $\alpha<1$ such that
\begin{equation}\label{eq:25}
\Re(\rho)>0,~~\Re(\frac{\eta}{1-\alpha})>-1.
\end{equation}
 Then
\begin{equation}\label{eq:26}
(P_{0+}^{(\eta, \alpha)}t^{\rho -1})(x)= \frac{\Gamma(\rho)
\Gamma(1+\frac{\eta}{1-\alpha})}{\Gamma(\frac{\eta}{1-\alpha}+\rho+1)}
\frac{x^{\eta+\rho}}{[a(1-\alpha)]^{\rho}},~x>0.
\end{equation}
In particular, for $x>0$
\begin{equation}\label{eq:a1}\displaystyle\lim_{\alpha\rightarrow 1_{-}}~(P_{0+}^{(\eta, \alpha)}t^{\rho -1})(x)\rightarrow
{\Gamma(\rho)}
\frac{x^{\eta+\rho}}{[a\eta]^{\rho}}.\end{equation}
\end{lemma}
We prove that such compositions are expressed in terms of the
 generalized Wright  hypergeometric function $_p\Psi_q(z)$ defined
 for $z\in \mathbb{C}$, complex $a_i,b_j \in\mathbb{C}$, and real $\alpha_i,\beta_j
 \in\mathbb{ R}( i=1,2,\ldots p; ~j=1,2,\ldots q)$ by the series
 \begin{eqnarray}\label{eq:5.25}
_p\Psi_q(z)&=&_p\Psi_q\bigl[_{(b_j,\beta_j)_{1,q}}^{(a_i,\alpha_i)_{1,p}}\big|z
\bigr]
\equiv\sum_{k=0}^{\infty}\frac{\prod_{i=1}^p\Gamma(a_i+\alpha_ik)}
{\prod_{j=1}^q\Gamma(b_j+\beta_jk)}\frac{z^k}{k!}.
\end{eqnarray}
Asymptotic behavior of this function for large values of argument of
$z$ was investigated by Fox \cite{Fox1928} and Wright \cite{Wright1}
under the condition
\begin{equation}
\sum_{j=1}^q\beta_j-\sum_{i=1}^p\alpha_i>-1.
\end{equation}
Under this conditions $_p\Psi_q(z)$ is an entire function, see
 Fox \cite{Fox1928}.

 The following assertion is based on the corresponding statement
for the generalized fractional integral (\ref{eq:4.1}) obtained in \cite{nair}.
\begin{theorem}
\label{th:1}
Let $\eta,\rho,b,p,c\in\mathbb{C}, \Re(1+\frac{\eta}{1-\alpha})>0,\Re(\rho+p)>0,\Re(\eta)>0, \alpha<1$ and $P_{0+}^{(\eta, \alpha)}$ be the pathway fractional integral. Then there holds the image.

\begin{eqnarray}\label{eq:5.26}
&&\left(P_{0+}^{(\eta, \alpha)}t^{\rho
-1}W_{p,b,c}(t)\right)(x)=\frac{x^{p+\rho+\eta}\Gamma(1+\frac{\eta}{1-\alpha})}{2^p[a(1-\alpha)]^{p+\rho}}\nonumber\\
& &\quad \times{_1\Psi_2}
\left[_{(\kappa,1),(\frac{\eta}{1-\alpha}+p+\rho+1,2)
}^{(p+\rho,2)}\big|-\frac{cx^2}{4[a(1-\alpha)]^2}
\right],
\end{eqnarray}

where $_p\Psi_q(z)$ is given by (\ref{eq:5.25})
and $\kappa=p+\frac{b+1}{2}$.
\end{theorem}
\begin{proof}
An application of integral operator (\ref{eq:4.1}) to the generalized Bessel function (\ref{eq:a}) leads to the formula
\begin{equation}\label{eq:5.27}
\left(P_{0+}^{(\eta, \alpha)}t^{\rho
-1}W_{p,b,c}(t)\right)(x)=\left(P_{0+}^{(\eta, \alpha)}\sum_{k=0}^\infty\frac{(-c)^k(1/2)^{\nu+2k}}{k!\Gamma(\kappa+k)}(t)^{\rho+p+2k-1}\right)(x).
\end{equation}
Now changing
the orders of integration and summation in the right hand side of (\ref{eq:5.27}) yields

\begin{equation}\label{eq:5.28}
\left(P_{0+}^{(\eta, \alpha)}t^{\rho
-1}W_{p,b,c}(t)\right)(x)=\sum_{k=0}^\infty\frac{(-c)^k(1/2)^{\nu+2k}}{k!\Gamma(\kappa+k)} \left(P_{0+}^{(\eta, \alpha)}(t)^{\rho+p+2k-1}\right)(x).
\end{equation}
Note that for any $k=0,1,\ldots$,
$\Re(\rho+p+2k)\geq\Re(\rho+p)>0$.
Applying Lemma~\ref{lem1} and replacing $\rho$
by $\rho+p+2k$, we obtain

\begin{eqnarray}\label{eq:5.29}
&&\left(P_{0+}^{(\eta, \alpha)}t^{\rho
-1}W_{p,b,c}(t)\right)(x)
=\frac{x^{p+\rho+\eta}\Gamma(1+\frac{\eta}{1-\alpha})}{2^p[a(1-\alpha)]^{p+\rho}}\nonumber\\
& &\quad \times
\sum_{k=0}^\infty\frac{\Gamma(\rho+p+2k)}{\Gamma(\kappa+k)\Gamma(\frac{\eta}{1-\alpha}+\rho+p+1+2k)}\frac{(-cx^2)^k}{{[4a^2(1-\alpha)^2]^k} k!}.
\end{eqnarray}
Interpreting the right hand side of (\ref{eq:5.29}), the equality (\ref{eq:5.26}) can be obtained from by using the definition of generalized Wright function.
\end{proof}

\begin{remark}
\label{re:1}
For $\alpha\rightarrow 1_{-},$ (\ref{eq:5.26}) gives the Laplace transform image:
$$\displaystyle\lim_{\alpha\rightarrow 1_{-}}~\left(P_{0+}^{(\eta, \alpha)}t^{\rho
-1}W_{p,b,c}(t)\right)(x)=
\frac{x^{p+\rho+\eta}}{2^p(a\eta)^{p+\rho}}~{_1\Psi_1}
\left[_{(\kappa,1)
}^{(p+\rho,2)}\big|-\frac{cx^2}{4(a\eta)^2}
\right].$$
\end{remark}

\begin{proof}
The Stirling's approximation for a gamma function, namely,
$${\Gamma (z+\beta)} \approx {{2\pi}^{\frac{1}{2}}
(z)^{z+\beta-\frac{1}{2}}{\rm e}^{-z}}$$ for $|z|\rightarrow \infty$ and $\beta$ a bounded quantity.
 In (\ref{eq:5.26}) when $\alpha\rightarrow 1_{-}, \frac{\eta}{1-\alpha}\rightarrow \infty$ and using the Stirling's formula of gamma functions gives the result.
%
%
\end{proof}

\subsection{Fractional integration of trigonometric functions.}
For all $b\in \mathbb{C}, p=-\frac{b}{2}$ then the generalized Bessel function $W_{p,b,c}(z)$ in
(\ref{eq:a}) coincides with the cosine function and hyperbolic cosine functions respectively by,
\begin{equation}\label{eq:5.45}
\begin{array}{lcl}
W_{-\frac{b}{2},b,c^2}(z)=\bigg(\frac{2}{ z}\bigg)^\frac{b}{2}\frac{ \cos
(c z)}{\sqrt{\pi}} & \text{and} &  W_{-\frac{b}{2},b,-c^2}(z)=\bigg(\frac{2}{ z}\bigg)^\frac{b}{2}\frac{ \cosh(c z)}{\sqrt{\pi}}.\end{array}
\end{equation}
Similarly for all $b\in \mathbb{C}, p=1-\frac{b}{2},$ then the generalized Bessel function $W_{p,b,c}(z)$ have the form
\begin{equation}\label{eq:5.46}
\begin{array}{lcl}
W_{1-\frac{b}{2},b,c^2}(z)=\bigg(\frac{2}{ z}\bigg)^\frac{b}{2}\frac{ \sin
(c z)}{\sqrt{\pi}} & \text{and} &  W_{1-\frac{b}{2},b,-c^2}(z)=\bigg(\frac{2}{ z}\bigg)^\frac{b}{2}\frac{ \sinh(c z)}{\sqrt{\pi}}.\end{array}
\end{equation}
From Theorem~\ref{th:1}  we obtain the following result:
\begin{corollary}
\label{co:1}
Let $\eta,\rho,b,c\in\mathbb{C}, \Re(1+\frac{\eta}{1-\alpha})>0,\Re(\rho)>0,\Re(\eta)>0, \alpha<1$ and $P_{0+}^{(\eta, \alpha)}$ be the pathway fractional integral. Then there holds the formula
\begin{eqnarray}\label{eq:5.47}
&&\left(P_{0+}^{(\eta, \alpha)}t^{\rho
-1}\cos(ct)\right)(x)=\frac{\sqrt{\pi}x^{\rho+\eta}\Gamma(1+\frac{\eta}{1-\alpha})}{[a(1-\alpha)]^{\rho}}\nonumber\\
& &\quad \times
{_1\Psi_2}
\left[_{(\frac{1}{2},1),(\frac{\eta}{1-\alpha}+\rho+1,2)
}^{(\rho,2)}\big|-\frac{c^2x^2}{4[a(1-\alpha)]^2}
\right],
\end{eqnarray}
and
\begin{eqnarray}\label{eq:5.48}
&&\left(P_{0+}^{(\eta, \alpha)}t^{\rho
-1}\cosh(ct)\right)(x)=\frac{\sqrt{\pi}x^{\rho+\eta}\Gamma(1+\frac{\eta}{1-\alpha})}{[a(1-\alpha)]^{\rho}}\nonumber\\
& &\quad \times
{_1\Psi_2}
\left[_{(\frac{1}{2},1),(\frac{\eta}{1-\alpha}+\rho+1,2)
}^{(\rho,2)}\big|\frac{c^2x^2}{4[a(1-\alpha)]^2}
\right].
\end{eqnarray}
\end{corollary}

\begin{remark}
\label{re:2}
For $\alpha\rightarrow 1_{-},$ (\ref{eq:5.47}) and (\ref{eq:5.48}) reduces to
$$\displaystyle\lim_{\alpha\rightarrow 1_{-}}~\left(P_{0+}^{(\eta, \alpha)}t^{\rho
-1}\cos(ct)\right)(x)\rightarrow
\frac{\sqrt{\pi}x^{\rho+\eta}}{(a\eta)^{\rho}}
{_1\Psi_1}
\left[_{(\frac{1}{2},1)
}^{(\rho,2)}\big|-\frac{c^2x^2}{4(a\eta)^2}
\right],$$
and
$$\displaystyle\lim_{\alpha\rightarrow 1_{-}}~\left(P_{0+}^{(\eta, \alpha)}t^{\rho
-1}\cosh(ct)\right)(x)\rightarrow
\frac{\sqrt{\pi}x^{\rho+\eta}}{(a\eta)^{\rho}}
{_1\Psi_1}
\left[_{(\frac{1}{2},1)
}^{(\rho,2)}\big|\frac{c^2x^2}{4(a\eta)^2}
\right].$$
\end{remark}

Thus from Theorem~\ref{th:1}, the composition of Pathway fractional integral operators respectively with sine and hyperbolic sine functions can be obtained.

\begin{corollary}
\label{co:2}
Let $\eta,\rho,b,c\in\mathbb{C}, \Re(1+\frac{\eta}{1-\alpha})>0,\Re(\rho)>0,\Re(\eta)>0, \alpha<1$ and $P_{0+}^{(\eta, \alpha)}$ be the pathway fractional integral. Then there holds the formula
\begin{eqnarray}\label{eq:5.49}
&&\left(P_{0+}^{(\eta, \alpha)}t^{\rho
-1}\sin(ct)\right)(x)=\frac{\sqrt{\pi}}{2}~\frac{x^{\rho+\eta+1}\Gamma(1+\frac{\eta}{1-\alpha})}{[a(1-\alpha)]^{\rho+1}}\nonumber\\
& &\quad \times
{_1\Psi_2}
\left[_{(\frac{3}{2},1),(\frac{\eta}{1-\alpha}+\rho+2,2)
}^{(\rho+1,2)}\big|-\frac{c^2x^2}{4[a(1-\alpha)]^2}
\right],
\end{eqnarray}
and
\begin{eqnarray}\label{eq:5.50}
&&\left(P_{0+}^{(\eta, \alpha)}t^{\rho
-1}\sinh(ct)\right)(x)=\frac{\sqrt{\pi}}{2}~
\frac{x^{\rho+\eta+1}\Gamma(1+\frac{\eta}{1-\alpha})}{[a(1-\alpha)]^{\rho+1}}\nonumber\\
& &\quad \times
{_1\Psi_2}
\left[_{(\frac{3}{2},1),(\frac{\eta}{1-\alpha}+\rho+2,2)
}^{(\rho+1,2)}\big|\frac{c^2x^2}{4[a(1-\alpha)]^2}
\right].
\end{eqnarray}
\end{corollary}

\begin{remark}
\label{re:3}
For $\alpha\rightarrow 1_{-},$ (\ref{eq:5.49}) and (\ref{eq:5.50}) reduces to
$$\displaystyle\lim_{\alpha\rightarrow 1_{-}}~\left(P_{0+}^{(\eta, \alpha)}t^{\rho
-1}\sin(ct)\right)(x)\rightarrow
\frac{\sqrt{\pi}}{2}~\frac{x^{\rho+\eta+1}}{(a\eta)^{\rho+1}}
{_1\Psi_1}
\left[_{(\frac{3}{2},1)
}^{(\rho+1,2)}\big|-\frac{c^2x^2}{4[a\eta]^2}
\right]$$
and
$$\displaystyle\lim_{\alpha\rightarrow 1_{-}}~\left(P_{0+}^{(\eta, \alpha)}t^{\rho
-1}\sinh(ct)\right)(x)\rightarrow
\frac{\sqrt{\pi}}{2}~\frac{x^{\rho+\eta+1}}{(a\eta)^{\rho+1}}
{_1\Psi_1}
\left[_{(\frac{3}{2},1)
}^{(\rho+1,2)}\big|\frac{c^2x^2}{4[a\eta]^2}
\right].$$
\end{remark}

\begin{example}
\label{ex:1}
We consider special case of (\ref{eq:5.47}), which gives the result in terms of the Mittag-Leffler
function.  When $\alpha=0, a=1$ and replace  $\eta$ by  $\eta-1$ we have
\begin{eqnarray}\label{eq:5.47a}
&&\left(I_{0+}^{\eta} t^{\rho
-1}\cos(ct)\right)(x)=\sqrt{\pi}x^{\rho+\eta-1}
{_1\Psi_2}
\left[_{(\frac{1}{2},1),(\eta+\rho,2)
}^{(\rho,2)}\big|-\frac{c^2x^2}{4}
\right].
\end{eqnarray}
For $\rho=c =1,$ relation
(\ref{eq:5.47a}) takes the form
\begin{eqnarray*}\label{eq:5.61}
(I_{0+}^{\eta} \cos t)(x)&=&\pi^{\frac{1}{2}}
x^{\eta}\
{_1\Psi_2}\left[_{(\frac{1}{2},1),(1+\eta,2)}
^{(1,2)}\bigg| \frac{-x^2}{4} \right].\nonumber\\
\end{eqnarray*}
Applying the formulas
\begin{equation}\label{eq:5.61a}\Gamma \left(\frac{1}{2}\right)=\pi^{1/2},\\ \ (2k)!=4^{k}k!\left(\frac{1}{2}\right)_k,\ \ k\in
\mathbb{N}_0,\end{equation}
we find
%
%
%
%
\begin{equation*}\label{eq:5.62}(I_{0+}^{\eta}\cos
t)(x)=x^{\eta}E_{2,1+\eta}\left(-x^2\right),
\end{equation*}
where $E_{\alpha,\beta}(\cdot)$ denotes the two-index Mittag-Leffler function and which is defined as $$E_{\alpha,\beta}(z)=\sum_{k=0}^\infty\frac{ z^k}
{\Gamma(\beta+\alpha k)}, z\in \mathcal{C}.$$
For positive integer $\eta =m\in \mathbb{N}$,
\begin{equation*}\label{eq:5.64}(I_{0+}^{m}\cos
t)(x)=\sum^{\infty}_{k=0} \frac{(-1)^kx^{2k+m}}{(2k+m)!},\ \ m\in
\mathbb{N}.\end{equation*}
In particular, for $m=1$ we obtain the well known formula
\begin{equation*}\label{eq:5.65}(I_{0+}^{1}\cos
t)(x)\equiv \int^{x}_{0}\cos t ~{\rm{d}}t =\sin x.
\end{equation*}
One can obtain similar kind of results in other cases too.
\end{example}

Thus these results are useful to derive ceratin composition formula involving Riemann-Liouville, Erd\'{e}lyi-Kober, Saigo and pathway fractional operators on Bessel, modified Bessel, and spherical Bessel function
of first kind. Particular attention is devoted to the technique of Laplace transform for
treating these operators in a way accessible to applied scientists, avoiding unproductive
generalities and excessive mathematical rigor.

\section{Applications.}

For the sake of
completeness we discuss the following application listed from Mathai and Haubold  \cite{hans}.

\subsection{Reaction-diffusion models.}
Reaction and relaxation processes in
thermonuclear plasmas are governed by ordinary differential equations of the type
\begin{equation}\label{eq:14}{\rm {\frac{d}{d t}}}N (t)={c}  N(t)\end{equation}
for exponential behavior. The quantity $c$ is a thermonuclear function which is governed by the average
of the Gamow penetration factor over the Maxwell-Boltzmannian velocity distribution of reacting
species and has been extended to incorporate more general distributions than the normal distribution
(\cite{ams2010}). To address non-exponential properties of a reaction or relaxation process,
the first-order time derivative can be replaced formally by a derivative of fractional order in the following
way (\cite{ams2010})
\begin{equation}\label{eq:15} N (t)=N_0 {-\rm c} ^{\nu}{_0D_t}^{-\nu}N (t),
\ \ \nu>0, \end{equation} where ${_0D_t}^{-\nu}f(t)$  is the Riemann-Liouville fractional integral operator, and the
solution can be written in terms of  Mittag-Leffler function:
\begin{equation}\label{eq:17}N(t)=N_0 \sum_{k=0}^\infty
\frac{(-1)^k[({c}  t)^\nu]^k}
{\Gamma(1+k\nu)}=N_0 E_\nu(-{c} ^{\nu} t^{\nu}).\end{equation}
The Laplace transform of $N (t)$
coming from (\ref{eq:15}) is
\begin{equation*}\label{eq:16}L_{N (t)}(s)=
\frac{N_0 }{s[1+(\frac{c }{s})^\nu]},\end{equation*}  which is a special case of general class of
Laplace transforms associated with $\alpha$-Laplace stochastic
processes and geometrically infinitely divisible statistical
distributions.

Considering $c$ to be a random variable itself, $N(t)$ is to be taken as$ N(t|c)$ and can be written as
$$N(t|c)=N_0 t^{\mu-1}E_{\nu,\nu}^{\gamma+1}(-{c} ^{\nu} t^{\nu}),\gamma,\mu,\nu>0,$$ which represents a generalized Mittag-Leffler function, and is a random variable having a gamma type density
$$g(c)=\frac{\omega^\mu}{\Gamma(\mu)}c^{\mu-1}{\rm e}^{-\omega c}, \omega, \mu>0, 0<c<\infty.$$ Hence the unconditional density will be
\begin{equation}\label{eq:15a} N (t)=\frac{N_0}{\Gamma(\mu)}t^{\mu-1}[1+b(\alpha-1)t^\nu]^{-\frac{1}{\alpha-1}},  \end{equation}
with $\gamma+1=1/(\alpha-1), \alpha>1 \rightarrow \gamma= (\alpha-2)/(\alpha-1)$ and $\omega^{-\nu}=b(\alpha-1), b>0,$ which
corresponds to Tsallis statistics for $\mu=\nu=b=1$ and $\alpha=q>1,$ physically meaning that the
common exponential behavior is replaced by a power-law behavior, including L\'{e}vy statistics. Both the
translation of the standard reaction equation (\ref{eq:14}) to a fractional reaction equation (\ref{eq:15}) and the probabilistic
interpretation of such equations lead to deviations from the exponential behavior to power law behavior
expressed in terms of Mittag-Leffler functions (\ref{eq:17}) or, as can be shown for equation (\ref{eq:15a}), to power
law behavior in terms of H-functions (\cite{mathaihaubold2007} , \cite{ams2010}).

\section*{Acknowledgement.}
Author acknowledges gratefully the encouragement given by Professor H. J. Haubold, Office of Outer Space Affairs, United Nations, Vienna International Centre, Austria.

%

\bibliography{bibliography}

\end{document}